\documentclass[10pt,a4paper,USenglish,preprint,3p]{elsarticle}
 
  \usepackage{amsmath,amsfonts}
  \usepackage{latexsym,amssymb}
  \usepackage[normalem]{ulem}
  \usepackage{amsthm}
  \usepackage{tikz}
  \usepackage[T1]{fontenc}
  \usepackage[utf8]{inputenc}
  \usepackage{comment}
  \usepackage{xspace}
  \usepackage{enumerate}
  \usepackage{listings}
  \usepackage[hidelinks]{hyperref}
  \usepackage[noline,noend]{algorithm2e}
  \usetikzlibrary{decorations.pathreplacing,calc,snakes} 
\usepackage{microtype}
\usepackage{lineno}
\usepackage[whileod]{chl-alg}

\theoremstyle{plain}
\newtheorem{theorem}{Theorem}

\newtheorem{lemma}[theorem]{Lemma}

\theoremstyle{definition}

\newtheorem{problem}{Problem}

\theoremstyle{remark}

\def\M{\mathcal{M}}

\begin{document}

\begin{frontmatter}

\title{Algorithms for Anti-Powers in Strings}

\author{Golnaz Badkobeh} 
\ead{golnaz.badkobeh@gmail.com}
\address{Department of Computer Science, University of Warwick, Warwick, UK}

\author{Gabriele Fici} 
\ead{gabriele.fici@unipa.it}
\address{Dipartimento di Matematica e Informatica, Universit\`a di Palermo, Italy}

\author{Simon J. Puglisi} 
\ead{puglisi@cs.helsinki.fi}
\address{
Helsinki Institute for Information Technology,\\
Department of Computer Science, University of Helsinki, Helsinki, Finland}

\sloppy  
  
\begin{abstract}
A string $S[1,n]$ is a power (or tandem repeat) of order $k$ and period $n/k$ if it can decomposed into $k$ consecutive equal-length blocks of letters. Powers and periods are fundamental to string processing, and algorithms for their efficient computation have wide application and are heavily studied. Recently, Fici et al. (Proc. ICALP 2016) defined an {\em anti-power} of order $k$ to be a string composed of $k$ pairwise-distinct blocks of the same length ($n/k$, called  {\em anti-period}). Anti-powers are a natural converse to powers, and are objects of combinatorial interest in their own right. In this paper we initiate the algorithmic study of anti-powers. Given a string $S$, we describe an optimal algorithm for locating all substrings of $S$ that are anti-powers of a specified order. The optimality of the algorithm follows form a combinatorial lemma that provides a lower bound on the number of distinct anti-powers of a given order: we prove that a string of length $n$ can contain $\Theta(n^2/k)$ distinct anti-powers of order $k$.
\end{abstract}


\end{frontmatter}


\section{Introduction}

A vast literature exists on algorithms for locating regularities in strings. One of the most natural notions of regularity is that of an exact repetition (also called power or tandem repeat), that is, a substring formed by two or more contiguous identical blocks --- the number of these identical blocks is called the \emph{order} of the repetition. 
Often, the efficiency of such algorithms derives from combinatorial results on the structure of the strings.
The reader is pointed to \cite{Ba15} for a survey on combinatorial results about text redundancies and algorithms for locating them. 

Recently, a new notion of regularity for strings based on diversity rather than on equality has been introduced: an \emph{anti-power}of order $k$ \cite{FRSZ16} (see~\cite{Fi18} for the extended version) is a string that can be decomposed into $k$ pairwise-distinct strings of identical length. This new notion is at the basis of a new unavoidable property. Indeed, regardless of the alphabet size, every infinite string must contain powers of any order or anti-powers of any order~\cite{FRSZ16,Fi18}. Defant~\cite{Def17} (see also Narayanan~\cite{N17}) studied the sequence of lengths of the shortest prefixes of the Thue-Morse word that are $k$-anti-powers, and proved that this sequence grows linearly in $k$.

In this paper, we focus on the problem of finding efficient algorithms to locate anti-powers in a finite string. While there exist several algorithms for locating repetitions in strings (see for example \cite{Cr09}), we present here the first algorithm that locates anti-power substrings in a given input string. Furthermore, we exhibit a lower bound on the number of distinct substrings that are anti-powers of a specified order, which allows us to prove that our algorithm time complexity is optimal.

\section{Preliminaries}

Let $S = S[1..n]$ be a string of length $|S|=n$ over an alphabet $\Sigma$ of 
size $|\Sigma|=\sigma$. 
The empty string $\varepsilon$ is the string of length $0$.
For $1 \leq i \le j \leq n$,
$S[i]$ denotes the $i$th symbol of $S$,
and $S[i..j]$ the contiguous sequence of symbols (called {\em factor} or {\em substring}) 
$S[i]S[i+1]\ldots S[j]$. 
A substring $S[i..j]$ is a suffix of $S$ if $j=n$ and it is a prefix of $S$ if $i=1$.
A {\em power of order $k$} (or {\em $k$-power}) is a string that is the concatenation of $k$ identical strings.
An {\em anti-power} of order $k$ (or {\em $k$-anti-power}) is a string that can be decomposed into $k$ pairwise-distinct strings of identical length \cite{FRSZ16}. The {\em period} of a $k$-power (resp.~the  {\em anti-period} of a $k$-anti-power) of length $n$ is the integer $n/k$.

For example, $S=aabaab$ is a $2$-power (also called a {\em square}) of period $3$, while $S=abcaba$ is a $3$-anti-power of anti-period $2$ (but also a $2$-anti-power of anti-period $3$).

In this paper, we consider the following problem:

\begin{problem}\label{Pm1}
Given a string $S$ and an integer $k>1$, locate all the substrings of $S$ that are anti-powers of order $k$.
\end{problem}

We describe an optimal solution to this problem in Section~\ref{sec:alg}. Before that, in Section~\ref{sec:lowerbound}, 
we prove a lower bound on the number of anti-powers of order $k$ that can be present in a string of length $n$, which allows
us to establish the optimality of our algorithm.

\section{Lower Bound on the Number of Anti-Powers}\label{sec:lowerbound}

Over an unbounded alphabet, it is easy to see that a string of length $n$ can contain $\Omega(n^2/k)$ anti-powers of order $k$ (think of a string consisting of all-distinct letters). However, somewhat more surprisingly, this bound also holds over a finite alphabet, as we now show. 

For every positive integer $m$, we let $w_m$ denote the string obtained by concatenating the binary expansions of integers from $0$ to $m$ followed by a symbol $\$$. So for example $w_5=0\$1\$10\$11\$100\$101\$$. We have that $|w_m|=\Theta(m\log m)$. Let us write $n=|w_m|$. 

\begin{lemma}
Every string $w_m$ of length $n$ contains $\Omega(\frac{n^2}{k})$ anti-powers of order $k$.
\end{lemma}

\begin{proof}
As mentioned before, we have $n=\Theta(m \log m)$. Let $AP(k,p)$ denote the number of anti-powers of order $k$ in $w_m$ with anti-period $p$.

The number of anti-powers of order $k$ is at least the sum of the number of anti-powers of order $k$ with anti-period greater than $3+2\lceil \log_2 m \rceil$. It is readily verified that if the anti-period $p$ is such that $p >3+2\lceil \log_2 m\rceil$ then at every position $i< n - pk$ in $w_m$ there is a $k$-anti-power of anti-period $p$. This is because there are at least two $\$$'s in every factor of $w_m$ of length $p>3+2\lceil \log_2 m\rceil$, 
and every factor of $w_m$ containing at least two $\$$'s 
has, by construction, only one occurrence in $w_m$.

Hence we have:
\begin{eqnarray*}
 \sum_{p>3+2\lceil \log_2 m \rceil}^{n/k} AP(k,p)
 &\geq& \sum_{p>3+2\lceil \log_2 m \rceil}^{n/k} (n-kp) \\
  &=& n\left(\frac{n}{k}-3-2\lceil \log_2 m \rceil\right)-k\left(\sum_{p=1}^{n/k}p-\sum_{p=1}^{3+2\lceil \log_2 m \rceil}p\right) \\
    &\geq & \frac{n^2}{k}
    -3n-2n\lceil \log_2 m \rceil
    -k\sum_{p=1}^{n/k}p \\
  &= & \frac{n^2}{k}-\frac{k}{2}\left(\frac{n}{k}\left(\frac{n}{k}+1\right)\right)-3n-2n\lceil \log_2 m \rceil \\
  &=& \frac{n^2}{k}-\frac{n^2}{2k}-\frac{n}{2}-3n-2n\lceil \log_2 m \rceil \\
  &=& \frac{n^2}{2k}-\frac{7n}{2}-2n\lceil \log_2 m \rceil.
 \end{eqnarray*}
 
 Thus we have $\sum_{p>3+2\lceil \log_2 m \rceil}^{n/k}  AP(k,p) = \Omega(\frac{n^2}{k})$, as claimed.
\end{proof}


\section{Computing Anti-Powers of Order $k$}\label{sec:alg}

This section is devoted to establishing the following theorem and we assume $S$ is over an alphabet $\Sigma=[n]$.

\begin{theorem}\label{Thm1}
Given a string $S[1,n]$ and an integer $k>1$, the locations of all substrings of $S$ that are $k$-anti-powers can be determined in $O(n^2/k$) time and $O(n)$ space.
\end{theorem}

In light of the lower bound established in the previous section on the number of anti-powers of a given order $k$ that can occur in a string, this solution to Problem~\ref{Pm1} is optimal.

\subsection{Computing anti-powers having anti-period $p=1$}

We begin with a lemma that we will use in our algorithm.

\begin{lemma}\label{lemma-alg-distinct-letters}
Given a string $S[1..n]$, the longest substring of $S$ that consists of pairwise-distinct symbols 
can be computed in $O(n)$ time and space.
\end{lemma}
\begin{proof}
We scan $S$ left to right, and maintain two pointers $x \le y$ into it. Through the scan,
both $x$ and $y$ are monotonically nondecreasing. We maintain the invariant that the symbols
in the substring delineated by $x$ and $y$, i.e., $S[x,y]$, are all distinct. In order to
maintain this invariant, we keep an array $P[1..\sigma]$, initially all 0s, such that immediately
before we increment $y$, $P[c] < y$ is the rightmost position of symbol $c$ in $S[1..y]$ (or
0 if $c$ does not appear in $S[1..y]$). Clearly, for the invariant to hold, we must have that
$P[S[y]] < x$, otherwise there are (at least) two occurrences of $S[y]$ in $S[x..y]$. In other
words, if $S[x..y]$ contains distinct letters then so will $S[x..y+1]$, provided $P[S[y+1]] < x$.
Initially, $x = y = 1$ and the invariant holds. We increment $y$ until $P[S[y]] > x$, at which point
we know that the symbols of $S[x..y-1]$ were distinct. If $S[x..y-1]$ is the length of the longest
such substring we have seen so far, we record $x$ and $y-1$. We then restore the invariant by setting
$x = P[S[y]]+1$, which has the effect of dropping the left occurrence of the repeated symbol $P[S[y]]$,
so that $S[x,y]$ again contains distinct symbols. The runtime is clearly linear in $n$. The only
non-constant space usage is for $P$.
\end{proof}

Obviously, the above algorithm can be used to efficiently compute $k$-anti-powers having anti-period 1. 
We will use it as a building block for finding $k$-anti-powers of all anti-periods.

\subsection{Optimal algorithm for computing anti-powers}

Let us now describe our algorithm. 
Firstly, observe that the maximum anti-period of a $k$-anti-power within $S$ is $p_{\mbox{\scriptsize max}} = n/k$.
Our algorithm works in $p_{\mbox{\scriptsize max}}$ rounds, $p = 1..p_{\mbox{\scriptsize max}}$.
In a generic round $p$ we will determine if $S$ contains (as a substring) a $k$-anti-power of anti-period $p$.
Let $M_{i,p}$ be an integer name for substring $S[i..i+p]$ 
amongst all substrings of length $p$ in $S$ --- two substrings $S[i..i+p]$ and $S[j..j+p]$ have the same name 
if and only if the substrings are equal. Note that the number of names for any substring length $p$ is bounded above
by $n$, the length of the string.
We can determine a suitable $M_{i,p}$ for all $i$ and $p$ 
in linear time from the names of substrings of length $p-1$ as follows.
We create an array of $n$ pairs, $(i,m)$, one for each position $i$ in the string. 
Initially, $m=0$ for all pairs.
In round $p=0..n/k$, we are computing the names of the substrings of length $p+1$.
We stably radix sort the pairs in $O(n)$ time using $S[i+p]$ as the sort
key for pair $(i,m)$. We then scan the sorted list of pairs, and for every run of adjacent pairs 
for which both $m$ and $S[i+p]$ are equal, we assign them the same new name $m'$, overwriting 
their $m$ fields. After this scan, clearly only substrings $S[i+p]$ and $S[j+p]$ of length $p$ that 
are equal will have the same name because they had the same $(p-1)th$ name and their last letters 
($S[i+p]$ and $S[j+p]$) are equal. We can now assign $M_{i,p}$ by scanning the list of pairs again 
and for each pair $(i,m)$ encountered setting $M_{i,p} \leftarrow m$.  

To find a $k$-anti-power of anti-period $p$, we must find a set of distinct $k$ substrings of length $p$,
whose starting positions are spaced exactly $p$ positions apart and so are all equal modulo $p$.

Let $X_r$ be the set of positions in $S$ that are equal to $r$ modulo $p$, i.e.,
$r = i\bmod p\ \forall i \in X_r$.

Let $\M^p_r$ be the string of length $|X_r| = \lceil n/p \rceil$ formed by concatenating
the $M_{i,p}$ values (in increasing order of $i$) for which $i \in X_r$. We can form
$\M^p_r$ in $O(n/p)$ time by visiting each $i \in X_r$ and computing $M_{i,p}$ in constant
time. As observed above, any substring of length $k$ in
$\M^p_r$ that contains all-distinct letters corresponds to a $k$-anti-power. In particular,
if $\M^p_r[i..i+k-1]$ is made up of distinct letters, then $S[(i-1)p+r..(k+i-1)p+r-1]$ is a $k$-anti-power.

Thus, in round $p$ of our algorithm we compute $\M^p_r$ for each $r = 1..p$. The total space and time required is
$O(n)$. We then scan each of these $\M^p_r$ strings in turn and detect substrings of length
$k$ containing distinct letters, using the algorithm in the proof of Lemma~\ref{lemma-alg-distinct-letters}.
This process is denoted by function {\sc Distinct}, in Line \ref{algo:6} of our Algorithm.  
Function {\sc Distinct} outputs a set of starting and ending positions of $k$-anti-powers whose anti-periods are $p$ and starting positions $i \mod p$.
The time required to scan each $\M^p_r$ string is $O(n/p)$ and so is $O(n)$ in total for
round $p$. The extra space needed for each scan is $O(n)$ for the array of previous positions.

Because each round takes $O(n)$ time, and there are $O(n/k)$ rounds, the total running time
to output all anti-powers of order $k$ is $O(n^2/k)$. Since we can reuse space between
rounds, the total space usage is $O(n)$.

\vspace{0.5cm}
\begin{algo}{{\sc AntiPowers}}{S,k}
\DOFORI{p}{1}{n/k}
\DOFORI{i}{1}{p}
 \SET{S'}{\Call{$\M^p_i$}{S}}
\label{algo:5}
 \SET{AP}{\Call{{\sc Distinct}}{S',k}}
 \label{algo:6}
 \RETURN{AP}
\OD
\OD
\end{algo}

\begin{table}
\begin{center}
\begin{tabular}{*{6}{c|}c}
$p$&1&2&2&3&3&3\\
\hline
$r$&1&1&2&1&2&3\\
\hline
$\M^p_r$&aabababbbabb&133434&22242&1263&245&434\\
\hline
$AP$&$\emptyset$&$\emptyset$&$\emptyset$&(1,9),(4,12)&(2,10)&(3,11)\\
\end{tabular}
\vspace{0.5cm}
\caption{The step-by-step computations performed by Algorithm \Algo{{\sc AntiPowers}} for input $S=aabababbbabb\$$ and $k=3$.}
\end{center}
\end{table}

\section{Conclusions and Open Problems}

The algorithm of the previous section is optimal in the sense that there are strings for which we must spend $\Theta(n^2/k)$
to simply list the antipowers of order $k$ because there are that many of them (as established in Section~\ref{sec:lowerbound}).
One wonders though if an output senstive algorithm is possible, one that takes, say, $O(n + c)$ time, where $c$ is the number 
of antipowers of order $k$ actually present in the input. Alternatively, do conditional lower bounds on antipower computation 
exist?

Many interesting algorithmic problems concerning anti-powers remain. For example, suppose we are to preprocess 
$S$ and build a data structure so that later, given queries of the form $(i, j, k)$, we have to determine quickly 
whether the substring $S[i..j]$ is an anti-power of order $k$. Using suffix trees~\cite{w1973} and weighted ancestor 
queries~\cite{GawrychowskiLN14} it is fairly straightforward to achieve $O(k)$ query time, in $O(n)$ space. Alternatively, 
by storing metastrings for all possible anti-periods, it is not difficult to arrive at a data structure that requires 
$O(n^2)$ space and answers queries in $O(1)$ time. Is it possible to achieve a space-time tradeoff between the extremes 
defined by these two solutions, or even better, to simultaneously achieve the minima of the space and query bounds?

\subparagraph*{Acknowledgements}

Our sincere thanks goes to the anonymous reviewers, whose comments materially improved our initial manuscript.
Golnaz Badkobeh is partially supported by the Leverhulme Trust on the Leverhulme Early Career Scheme. 
Simon J. Puglisi is supported by the Academy of Finland via grant 294143.

\bibliographystyle{plain}
\bibliography{anti}
\end{document}